\documentclass[conference]{IEEEtran}

\usepackage{amsfonts}
\usepackage{amssymb}
\usepackage{mathrsfs}
\usepackage{amsmath}
\usepackage{cite}
\usepackage{mathrsfs}
\usepackage[ruled,vlined]{algorithm2e}
\usepackage{tikz}
\usetikzlibrary{arrows,automata}
\usepackage[latin1]{inputenc}
\usepackage{verbatim}
\makeatletter

\newcommand{\Rmnum}[1]{\expandafter\@slowromancap\romannumeral #1@}
\makeatother

\newtheorem{thm}{Theorem}
\newtheorem{corollary}[thm]{Corollary}
\newtheorem{lemma}[thm]{Lemma}
\newtheorem{eg}{Example}

\newtheorem{defn}{Definition}
\newtheorem{rem}{Remark}

\newtheorem{rem-eg}[thm]{Remark and Example}

\newcommand{\mP}{\mathcal{P}}
\newcommand{\mF}{\mathcal{F}}
\newcommand{\Rank}{{\mathrm{Rank}}}

\newenvironment{proof}[1][Proof]{\noindent \textbf{#1.} }{\hfill$\Box$}

\begin{document}

\title{The Failure Probability at Sink Node of Random Linear
Network Coding}

\author{\IEEEauthorblockN{Xuan Guang}
\IEEEauthorblockA{Chern Institute of
Mathematics\\Nankai University, P. R. China\\
Email: xuanguang@mail.nankai.edu.cn}
\and
\IEEEauthorblockN{Fang-Wei Fu}
\IEEEauthorblockA{Chern Institute of
Mathematics and LPMC\\
Nankai University, P. R. China\\
Email: fwfu@nankai.edu.cn}}

\maketitle

\begin{abstract}
In practice, since many communication networks are huge in scale or
complicated in structure even dynamic, the predesigned network
codes based on the network topology is impossible
even if the topological structure is known. Therefore, random
linear network coding was proposed as an acceptable coding
technique. In this paper, we further study the performance of random
linear network coding by analyzing the failure probabilities at sink
node for different knowledge of network topology and get some tight and
asymptotically tight upper bounds of the failure probabilities. In
particular, the worst cases are indicated for these bounds. Furthermore, if
the more information about the network topology is utilized,
the better upper bounds are obtained. These bounds improve on the
known ones. Finally, we also discuss the lower bound of this failure
probability and show that it is also asymptotically tight.
\end{abstract}

\IEEEpeerreviewmaketitle

\section{Introduction}
Network coding was first introduced by Yeung and Zhang in \cite{Zhang-Yeung-1999} and then profoundly developed in Ahlswede \emph{et al.} in \cite{Ahlswede-Cai-Li-Yeung-2000}. In the latter paper
\cite{Ahlswede-Cai-Li-Yeung-2000}, the authors showed that if coding
is applied at the nodes instead of routing alone, the source node
can multicast the information to all sink nodes at the theoretically
maximum rate. Li \emph{et al.} \cite{Li-Yeung-Cai-2003} indicated
that linear network coding with finite alphabet size is sufficient
for multicast. In \cite{Koetter-Medard-algebraic}, Koetter and M$\acute{\textup{e}}$dard presented an algebraic
characterization of network coding. Although network coding allows the
higher information rate than classical routing, Jaggi \emph{et al.}
\cite{co-construction} still proposed a deterministic
polynomial-time algorithm to construct a linear network code. For a
detail and comprehensive discussion of network coding, refer to
\cite{Zhang-book}, \cite{Yeung-book}, \cite{Fragouli-book},
\cite{Fragouli-book-app}, and \cite{Ho-book}.

Random linear network coding was originally proposed and analyzed
in the papers Ho \emph{et al.} \cite{Ho-random-conference-paper} and
\cite{Ho-etc-random}, where the main results are upper bounds for
failure probabilities of the code. Balli, Yan, and Zhang
\cite{zhang-random} improved on these bounds and the tightness of the
new bounds was studied by analyzing the asymptotic behavior of the
failure probability as the field size goes to infinity. However, the
upper bounds of failure probabilities proposed either by Ho \emph{et al.} \cite{Ho-etc-random} or by Balli \emph{et al.}
\cite{zhang-random} are not tight. In this paper, we further study
the random linear network coding and improve on the bounds of the
failure probabilities for different cases. In particular, if the
more knowledge about the topology of the network is known, we can
get the better bounds. Further, we indicate that these bounds
are either tight or asymptotically tight.

\section{Linear Network Coding and Preliminaries}
A communication network is defined as a finite acyclic directed
graph $G=(V,E)$, where the vertex set $V$ stands for the set of the
nodes and the edge set $E$ represents the set of communication
channels of the network. The nodes set $V$ consists of three
disjoint subsets $S$, $T$, and $J$, where $S$ is the set of source
nodes, $T$ is the set of sink nodes, the other nodes in
$J=V-S-T$ are called internal nodes and thus the subset $J$ is called the
set of internal nodes. A direct edge $e=(i,j)\in E$
represents a channel leading from node $i$ to node $j$. Node $i$
is called the tail of the channel $e$, node $j$ is called the head of the
channel $e$, and they are written as $i=tail(e)$, $j=head(e)$,
respectively. Correspondingly, the channel $e$ is called an outgoing
channel of $i$ and an incoming channel of $j$. For each node $i$,
define 
\begin{align*}
Out(i)=\{e\in E:\ e \mbox{ is an outgoing channel of }i\},\\
In(i)=\{e \in E:\ e \mbox{ is an incoming channel of }i\}.
\end{align*} 
For each
channel $e\in E$, there exists a positive number $R_e$ called the
capacity of the channel $e$. We allow the multiple channels between
two nodes and then assume reasonably that all capacity of the
channel is unit 1. That is, one field symbol can be transmitted over
a channel in one unit time. The source nodes generate messages and
transmit them to all sink nodes over the network by network coding.

In this paper, we sequentially consider single source multicast
networks, i.e. $|S|=1$, and the unique source node is denoted by $s$.
The source node $s$ has no incoming channels and any sink node
has no outgoing channels, but we use the concept of the imaginary
incoming channels of the source node $s$ and assume that these
imaginary channels provide the source messages to $s$. Let the
information rate be $w$ symbols per unit time which means that the source node
$s$ has $w$ imaginary incoming channels $d_1,d_2,\cdots,d_w$ and let $In(s)=\{d_1,d_2,\cdots,d_w\}$. The
source messages are $w$ symbols
$\underline{\bf{X}}=(X_1,X_2,\cdots,X_w)$ arranged in a row vector
where each $X_i$ is an element of the finite base field
$\mathcal{F}$. Assume that they are transmitted to $s$ through the $w$ imaginary channels. Using network
coding, these messages are multicast to each sink node and decoded
at each sink node.

We use $U_e$ to denote the message transmitted over channel $e=(i,j)$ and
$U_e$ is calculated by the following formula
$$U_e=\sum_{d\in In(i)}k_{d,e}U_d\ ,$$
where at the source node $s$, assume that the message transmitted
over $i$th imaginary channel $d_i$ is the $i$th source message, i.e.
$U_{d_i}=X_i$. And, by the definition of the global kernels of
the channel $e$, we have $U_e=\underline{\bf{X}}\cdot f_e$.

The linear network coding discussed above was designed based on
the global topology of the network. However, in most
communication networks, we cannot utilize the global topology
because the network is huge in scale, or complicated in
structure, even dynamic, or some another reasons. In other words, it is
impossible to use the predesigned codes based on the global
topology. Thus random linear network coding
was proposed as an acceptable coding technique. The main idea of random network
coding is that when a node (may be the source node $s$) receives the messages from its all
incoming channels, for each outgoing channel, it randomly and
uniformly picks the encoding coefficients from the base field
$\mathcal{F}$, uses them to encode the messages, and transmits the
encoded messages over the outgoing channel. In other words, the
local coding coefficients $k_{d,e}$ are independently and uniformly
distributed random variables in the base field $\mF$. Since random
linear network coding does not consider the network global topology
or does not coordinate codings at different nodes, it may not
achieve the best possible performance of network coding, that is,
some sink nodes may not decode correctly. Therefore, the performance
analysis of random linear network coding is important in theory
and application.

Before further discussion, we introduce some notation and
definitions as follows.

Let $A$ be a set of vectors from a linear space. $\langle A \rangle$
represents a linear subspace spanned by the vectors in $A$. In
addition, we give the definition of the failure probability at sink
node which was introduced exactly in \cite{zhang-random}.
\begin{defn}
Let $G$ be a single source multicast network, and the information rate
be $w$ symbols per unit time. $P_{e_t}\triangleq Pr(\Rank(F_t)<w)$
is called the failure probability of the random linear network
coding at sink node $t$, that is the probability that the source
messages cannot be decoded correctly at sink node $t\in T$.
\end{defn}

\section{Failure Probabilities of Random Linear Network Coding at Sink Node}
We have known that the performance analysis of random linear network
coding is very important in theory and application. In particular,
the random linear network coding is an acceptable coding technique
for non-coherent networks. However, many coherent networks
are huge and complicated, and thus the random linear network coding are often
used for the coherent networks. In this section, we study the
failure probability $P_{e_t}$ from coherent to non-coherent
networks. At first, we give the following lemma.

\begin{lemma}{\label{lem_bound}}
Let $\mathcal{L}$ be a $n$-dimensional linear space over finite
field $\mathcal{F}$, $\mathcal{L}_0,\ \mathcal{L}_1$ be two
subspaces of $\mathcal{L}$ of dimensions $k_0,\ k_1$,
respectively, and
$\langle\mathcal{L}_0\cup\mathcal{L}_1\rangle=\mathcal{L}$. Let
$l_1,\ l_2,\ \cdots,\ l_m$ $(m=n-k_0)$ be $m$ independently
uniformly distributed random vectors taking values in
$\mathcal{L}_1$. Then
$$Pr(\dim(\langle \mathcal{L}_0 \cup \{l_1,\
l_2,\ \cdots,\ l_m\}\rangle)=n)=\prod_{i=1}^{n-k_0}\left(
1-\frac{1}{\mathcal{|F|}^i}\right). $$
Therefore,
$$\frac{1}{\mathcal{|F|}}\leq Pr(\dim(\langle \mathcal{L}_0 \cup \{l_1,\
l_2,\ \cdots,\ l_m\}\rangle)<n)<\frac{1}{\mathcal{|F|}-1}.$$
\end{lemma}

\begin{rem}
We can observe that under the condition of Lemma
\ref{lem_bound}, $Pr(\dim(\langle \mathcal{L}_0 \cup
\{l_1,l_2,\cdots,l_m\}\rangle)=n)$ is not related to the dimension
of $\mathcal{L}_1$.
\end{rem}

Let $G$ be a single source multicast network, where the single
source node is denoted by $s$, the set of the sink nodes is denoted
by $T$, and the minimum cut capacity between $s$ and $t\in T$ is
$C_t$. The information rate is $w\leq\min_{t\in T}C_t$ symbols per
unit time.

For each sink node $t\in T$, since $w\leq C_t$ and Menger's
Theorem, there exist $w$ channel-disjoint paths from $s$ to $t$.
Let the arbitrarily chosen $w$ channel-disjoint paths from $s$ to
$t$ be $\mP_t=\{P_{t,1},P_{t,2},\dots,P_{t,w}\}$ and let
$P_{t,i}=\{e_{i,1},e_{i,2},\cdots,e_{i,m_i}\}$ satisfying
$tail(e_{i,1})=s,\ head(e_{i,m_i})=t$, and
$head(e_{i,j-1})=tail(e_{i,j})$ for others. The set of all channels
in $\mathcal{P}_t$ is denoted by $E_{\mathcal{P}_t}$. Furthermore,
assume that the number of the nodes in $\mathcal{P}_t$ is $r+2$,
where one is the source node $s$, one is the sink node $t$, and
another $r$ are internal nodes, which are denoted
by $i_1,i_2,\cdots,i_r$. There is a topological order ancestrally,
and without loss of generality, let the order be
$$s\triangleq i_0\prec i_1 \prec i_2 \prec \dots \prec i_r \prec i_{r+1}\triangleq t\ .$$

During our discussion, we use the concept of cuts of the paths from
$s$ to $t$ proposed in \cite{zhang-random}, which is different from
the concept of cuts of the network in graph theory. The first cut is
$CUT_{t,0}=In(s)$, i.e. the set of the $w$ imaginary channels.
Through the node $i_0=s$, the next cut $CUT_{t,1}$ is the set of the
first channels of all $w$ paths, i.e. $CUT_{t,1}=\{e_{i,1}:\ 1\leq i
\leq w\}$. Through the node $i_1$, the next cut $CUT_{t,2}$ is
formed from $CUT_{t,1}$ by replacing those channels in $In(i_1)\cap
CUT_{t,1}$ by their respective next channels in the paths. These new
channels are in $Out(i_1)\cap E_{\mathcal{P}_t}$. Other channels
remain the same as in $CUT_{t,1}$. Subsequently, once $CUT_{t,k}$ is
defined, $CUT_{t,k+1}$ is formed from $CUT_{t,k}$ by the same method
as above. By induction, all cuts $CUT_{t,k}$ for $k=0,1,\cdots,r+1$
can be defined. Furthermore, for each $CUT_{t,k}$, we divide
$CUT_{t,k}$ into two disjoint parts $CUT_{t,k}^{out}$ and
$CUT_{t,k}^{in}$, where
\begin{align*}
CUT_{t,k}^{out}=\{e:\ e\in CUT_{t,k}\setminus In(i_{k})\},\\
CUT_{t,k}^{in}=\{e:\ e\in CUT_{t,k}\cap In(i_{k})\}.
\end{align*}

\begin{thm}\label{thm_cohe_sink_original}
For this network $G$ mentioned as above, the failure probability of
random linear network coding at sink node $t\in T$ satisfies
$$P_{e_t}\leq 1-\prod_{k=0}^{r}\prod_{i=1}^{w-|CUT_{t,k}^{out}|}\left(1-\frac{1}{\mathcal{|F|}^i}\right)\ .$$
\end{thm}
\begin{proof}
For sink node $t\in T$, the decoding matrix $F_t=\begin{pmatrix}f_e:e\in In(t)\end{pmatrix}$
is a $w\times |In(t)|$ matrix over the field $\mF$. Define a
$w\times w$ matrix
$F_t'=\begin{pmatrix}f_{e_{1,m_1}},f_{e_{2,m_2}},\cdots,f_{e_{w,m_w}}\end{pmatrix}$. It is not
hard to see that $F_t'$ is a submatrix of $F_t$. It follows that the
event `` $\Rank(F_t)<w$ '' $\subseteq$ the event `` $\Rank(F_t')<w$
''. This means that
$$Pr(\Rank(F_t)<w)\leq Pr(\Rank(F_t')<w)\ .$$

Further define $w\times w$ matrices $F_t^{(k)}=(f_e: e\in
CUT_{t,k})$ for $k=0,1,\cdots, r+1$. If $\Rank(F_t^{(k)})<w$, we
call that we have a failure at $CUT_{t,k}$. we use $\Gamma_{t,k}$ to denote the event ``$\Rank(F_t^{(k)})=w$''. Obviously,
$F_t^{(r+1)}=F_t'$ because
$CUT_{t,r+1}=\{e_{1,m_1},e_{2,m_2},\cdots,e_{w,m_w}\}$. This implies
\begin{align*}
&Pr(\Rank(F_t')<w)=Pr(\Rank(F_t^{(r+1)})<w)\\
=&Pr((\Gamma_{t,r+1})^c)=1-Pr(\Gamma_{t,r+1}).
\end{align*}
In addition, since encoding at any node is independent of what
happened before this node as long as no failure has occurred up to
this node, we have
\begin{align}
&Pr(\Gamma_{t,r+1})\geq Pr(\Gamma_{t,r+1}\Gamma_{t,r}\cdots
\Gamma_{t,1}\Gamma_{t,0})\nonumber\\
=&Pr(\Gamma_{t,r+1}|\Gamma_{t,r})Pr(\Gamma_{t,r}|\Gamma_{t,r-1})\cdots
Pr(\Gamma_{t,1}|\Gamma_{t,0})Pr(\Gamma_{t,0})\nonumber\\
=&Pr(\Gamma_{t,r+1}|\Gamma_{t,r})Pr(\Gamma_{t,r}|\Gamma_{t,r-1})\cdots
Pr(\Gamma_{t,1}|\Gamma_{t,0})\label{pr=1}
\end{align}
where (\ref{pr=1}) follows because
$Pr(\Gamma_{t,0})=Pr(\Rank((f_e:e\in In(s)))=w)=Pr(\Rank(I_{w\times
w})=w)\equiv1$ with $I_{w\times w}$ being $w\times w$ identity matrix.

Therefore, applying Lemma \ref{lem_bound} for each $k$\ ($0\leq k \leq
r)$, we have
\begin{equation}
Pr(\Gamma_{t,k+1}|\Gamma_{t,k})=
\prod_{i=1}^{w-|CUT_{t,k}^{out}|}\left(1-\frac{1}{\mathcal{|F|}^i}\right)\label{thm1.1}\\
\end{equation}
where under the condition $\Gamma_{t,k}$, there must be
$|CUT_{t,k}^{out}|=\dim(\langle \{f_e: e\in
CUT_{t,k}^{out}\}\rangle)=\Rank((f_e: e \in CUT_{t,k}^{out}))$.

Combining (\ref{pr=1}) and (\ref{thm1.1}), it follows that
$$Pr(\Gamma_{t,r+1})\geq \prod_{k=0}^{r}\prod_{i=1}^{w-|CUT_{t,k}^{out}|}\left(1-\frac{1}{\mathcal{|F|}^i}\right).$$
That is, we get the upper bound of the failure probability at the
sink node $t$,
$$P_{e_{t}}\leq1-\prod_{k=0}^{r}\prod_{i=1}^{w-|CUT_{t,k}^{out}|}\left(1-\frac{1}{\mathcal{|F|}^i}\right)\ .$$
The proof is completed.
\end{proof}

\begin{rem}
This upper bound of the failure probability at the sink node $t$ in
Theorem $\mathrm{\ref{thm_cohe_sink_original}}$ is tight.
\end{rem}

\begin{eg}
For the well-known butterfly network, by Theorem
\ref{thm_cohe_sink_original} we know
$$P_{e_{t}}\leq 1-\prod_{i=1}^{2}(1-\frac{1}{|\mF|^{i}})(1-\frac{1}{|\mF|})^4
=1-\frac{(|\mF|+1)(|\mF|-1)^6}{|\mF|^7}\ .$$ On the other hand,
Guang and Fu \cite{Guang-Fu-ButterflyNet} have shown that for the
butterfly network $P_{e_{t}}=1-(|\mF|+1)(|\mF|-1)^6/|\mF|^7\ .$ This means that this upper bound is tight for the butterfly network.
\end{eg}

However, this upper bound is too complicated in practice. Thus, we
have to give a simpler in form but looser upper bound.
\begin{thm}\label{thm_cohe_sink}
For this network $G$, the failure probability of the random linear
network coding at sink node $t\in T$ satisfies
$$P_{e_t}\leq 1-\left[\prod_{i=1}^{w}\left(1-\frac{1}{|\mathcal{F}|^i}\right)\right]^{r+1} .$$
\end{thm}

In particular, if we choose the $w$ channel-disjoint paths with the
minimum number of the internal nodes among the collection of all $w$
channel-disjoint paths from $s$ to $t$ over network $G$, and denote
this minimum number by $R_t$, then we get a smaller upper bound
with the same simple form.
\begin{corollary}\label{cor_co_sink}
For this network $G$, the failure probability of the random linear
network coding at sink node $t\in T$ satisfies
$$P_{e_{t}}\leq1-\left[\prod_{i=1}^{w}\left(1-\frac{1}{|\mF|^i} \right)\right]^{R_t+1}.$$
\end{corollary}

\begin{rem}
Both upper bounds of the failure probability at the sink node in
Theorem $\mathrm{\ref{thm_cohe_sink}}$ and Corollary
$\mathrm{\ref{cor_co_sink}}$ are tight, and we can show the
tightness by the same way. Therefore, we only construct a network to
show the tightness of the upper bound in Theorem
$\mathrm{\ref{thm_cohe_sink}}$. In other words, we will give a
network as the worst case.
\end{rem}

\begin{eg}\rm\
\begin{figure}[!htb]
\begin{center}
\begin{tikzpicture}[->,>=stealth',shorten >=1pt,auto,node distance=2cm,
                    thick]
  \tikzstyle{every state}=[fill=none,draw=black,text=black,minimum size=6mm]
  \tikzstyle{place}=[fill=none,draw=white,minimum size=0.1mm]
  \node[state]         (s)                 {$s$};
  \node[state]         (i_1)[right of=s]   {$i_1$};
  \node[state]         (i_2)[right of=i_1]   {$i_2$};
  \node[state]         (i_r)[right of=i_2]   {$i_r$};
  \node[state]         (t)[right of=i_r]     {$t$};

  \path (s) edge [bend left=45]          node {$w$ channels} (i_1)
            edge                         node {$\vdots$} (i_1)
            edge [bend right=45]         node {} (i_1)
      (i_1) edge [bend left=45]          node {$w$ channels} (i_2)
            edge                         node {$\vdots$} (i_2)
            edge [bend right=45]         node {} (i_2)
      (i_2) edge [-,dashed]                node {} (i_r)
      (i_r) edge [bend left=45]          node {$w$ channels} (t)
            edge                         node {$\vdots$} (t)
            edge [bend right=45]         node {} (t);
\end{tikzpicture}
\caption{Plait Network with $r$ internal nodes}
\label{fig_plait_net}
\end{center}
\end{figure}
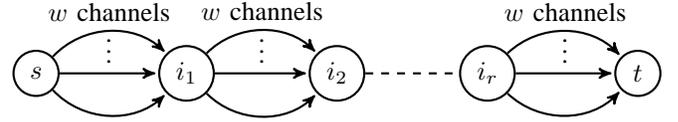
For the given information rate $w$, the network $G_1$ shown by
Fig.\ref{fig_plait_net} can be constructed as follows. Let the
source node be $s$, the sink node be $t$, the number of the internal
nodes be $r$, and denote these internal nodes by
$i_1,i_2,\cdots,i_{r}$. Let the topological order of all nodes be
$$s \prec i_1\prec i_2 \prec  \cdots \prec i_{r} \prec t\ .$$

Draw $w$ parallel channels from $s$ to $i_1$, $w$ parallel channels
from $i_1$ to $i_2$, in succession, $w$ parallel channels from
$i_{r}$ to $t$. The total $(r+1)w$ channels are all channels of the
network $G_1$. For this type of networks, we call them plait
networks. For this constructed network $G_1$, we will show that the
failure probability $P_{e_t}$ at sink node $t$ is
$$P_{e_t}= 1-\left[\prod_{i=1}^{w}\left(1-\frac{1}{|\mathcal{F}|^i}\right)\right]^{r+1}\ \ . $$

It is not difficult to see that the event ``$\Rank(F_t)<w$'' is
equivalent to the event ``$\Rank(F_t^{(r+1)})<w$'' because of
$F_t=F_t^{(r+1)}$. This implies
$$P_{e_t}=Pr(\Rank(F_t^{(r+1)})<w)=1-Pr(\Gamma_{t,r+1}).$$
Furthermore, for $G_1$,
\begin{align*}
&Pr(\Gamma_{t,r+1})=Pr(\Gamma_{t,r+1}\Gamma_{t,r}\cdots\Gamma_{t,1}\Gamma_{t,0})\\
=&Pr(\Gamma_{t,r+1}|\Gamma_{t,r})Pr(\Gamma_{t,r}|\Gamma_{t,r-1})\cdots Pr(\Gamma_{t,1}|\Gamma_{t,0}).
\end{align*}
And, for any $k=0,1,\cdots, r$,
\begin{equation*}
\begin{split}
&Pr(\Gamma_{t,k+1}|\Gamma_{t,k})\\
=&Pr(f_{e_{k,1}}\notin \langle
\underline{0} \rangle,f_{e_{k,2}}\notin \langle f_{e_{k,1}}\rangle,f_{e_{k,3}}\notin \langle f_{e_{k,1}},f_{e_{k,2}}\rangle\cdots,\\
&f_{e_{k,w}}\notin \langle \{f_{e_{k,1}},\ \cdots,\ f_{e_{k,{w-1}}}\} \rangle)\\
=&\prod_{i=1}^w\left(1-\frac{1}{\mathcal{|F|}^{i}}\right),
\end{split}
\end{equation*}
where $In(i_{k+1})=Out(i_{k})=\{e_{k,1},e_{k,2},\cdots,e_{k,w}\}$
and $\underline{0}$ is a zero vector.

Combining the above, we get
$$Pr(\Gamma_{t,r+1})=\left[ \prod_{i=1}^w\left(1-\frac{1}{\mathcal{|F|}^{i}}\right) \right]^{r+1}\ ,$$
that is,
$$P_{e_t}=1-Pr(\Gamma_{t,r+1})
=1-\left[ \prod_{i=1}^w\left(1-\frac{1}{\mathcal{|F|}^{i}}\right)
\right]^{r+1}.$$ This means that the upper bound of the failure
probability at the sink node is tight, and the type of plaint
networks is the worst case.
\end{eg}

As mentioned above, sometimes, it is hard to use the predesigned
linear network coding based on the network topology even through the topology of the network is
known. But usually we still can get some information about the
network topology more or less. For instance, we
can know the number of the internal nodes $|J|$ at least. In these
cases, we also can analyze the performance of the random linear
network coding.

\begin{thm}\label{thm_nco_sink}
Let $G$ be a single source multicast network. Using the random linear
network coding, the failure probability at the sink node $t\in
T$ satisfies
$$P_{e_t}\leq 1-\left[\prod_{i=1}^{w}\left(1-\frac{1}{|\mathcal{F}|^i}\right)\right]^{|J|+1}\ .$$
\end{thm}

\begin{rem}\label{rem_nco_sink}
This upper bound is still tight and we can also give an example to
indicate the tightness.
\end{rem}
\begin{eg}
For a given information rate $w$, construct a plait network $G_2$,
where the unique source node is $s$, the sink node is $t$, and all
internal nodes are $i_1,i_2,\cdots,i_{|J|}$. Let the topological
order of all nodes be
$$s\triangleq i_0 \prec i_1\prec i_2 \prec  \cdots \prec i_{|J|} \prec i_{|J|+1}\triangleq t\ .$$
There are $w$ parallel channels from $i_j$ to $i_{j+1}$, $0\leq j
\leq |J|$. Similar to the example above, we obtain that the failure
probability $P_{e_t}$ for plait network $G_2$ is
$$P_{e_t}= 1-\left[\prod_{i=1}^{w}\left(1-\frac{1}{|\mathcal{F}|^i}\right)\right]^{|J|+1}. $$
\end{eg}

\section{The Lower Bounds of The Failure Probabilities}
In the last section, we give some upper bounds of the failure
probability at sink node in order to analyze performance of random
linear network coding. In this section, we give the lower bound of
this failure probability.

\begin{thm}
For a single source multicast network $G$, using random linear
network coding, the failure probability at the sink node
satisfies $P_{e_t}\geq 1/|\mF|^{\delta_t+1}$, where
$\delta_t=C_t-w$.
\end{thm}

\begin{rem}
The lower bound in this theorem is also asymptotically tight.
\end{rem}
\section{Conclusion}
The performance of random linear network coding is important for
theory and application. In the present paper, we further analyze the upper
bounds of failure probability at sink node. In particular, the
more information about the network topology is utilized, the
better upper bounds are obtained. We further discuss the
lower bound of this failure probability and indicate that it is also
asymptotically tight.

In addition, other probabilities, such as failure probability for network and average failure probability, can also be defined to characterize the performance of random linear network coding. We have also analyzed these probabilities. But due to limited pages, we omit them.


\end{document}